\documentclass[10pt,conference]{IEEEtran}
\usepackage{amsmath,graphicx}

\usepackage{ifpdf}
\usepackage{manfnt}
\usepackage[mathscr]{eucal}
\usepackage{graphicx}

\usepackage{algorithmicx}
\usepackage[ruled,vlined,commentsnumbered]{algorithm2e}
\usepackage{amssymb}
\usepackage{amsthm}
\usepackage{caption}
\usepackage{subcaption}
\usepackage{cite}
\usepackage{color}
\usepackage{colortbl}
\definecolor{colorref}{rgb}{0.4648,0,0} 
\definecolor{colorcite}{rgb}{0,0.2902,0.1765}

\usepackage{amsmath}
\newtheorem{proposition}{Proposition}

\newtheorem{Cor}{Corollary}

\newcommand{\dl}{\mathrm{d}}
\newcommand{\ul}{\mathrm{u}}

\newcommand{\signal}[1]{{\boldsymbol{#1}}}
\newcommand{\real}{{\mathbb R}}
\newcommand{\innerprod}[2]{\left\langle{#1},{#2}\right\rangle}

\newtheorem{fact}{Fact}

\newcommand{\refeq}[1]{(\ref{#1})}

\newcommand{\dprime}{{\prime\prime}}

\begin{document}

\title{Error Bounds for FDD Massive MIMO Channel Covariance Conversion with Set-Theoretic Methods}

\author{\IEEEauthorblockN{Renato L. G. Cavalcante,$\null^\dagger$ L. Miretti,$\null^\ddagger$ S. Sta\'nczak$\null^\dagger$\\ 
		$\null^\dagger$ Fraunhofer Heinrich Hertz Institute and Technical University of Berlin, Germany \\ 
	$\null^\ddagger$ EURECOM, France	}
}

\maketitle

\begin{abstract}
We derive novel bounds for the performance of algorithms that estimate the downlink covariance matrix from the uplink covariance matrix in frequency division duplex (FDD) massive multiple-input multiple-output (MIMO) systems. The focus is on algorithms that use estimates of the angular power spectrum as an intermediate step. Unlike previous results, the proposed bounds follow from simple arguments in possibly infinite dimensional Hilbert spaces, and they do not require strong assumptions on the array geometry or on the propagation model. Furthermore, they are suitable for the analysis of set-theoretic methods that can efficiently incorporate side information about the angular power spectrum. This last feature enables us to derive simple techniques to enhance set-theoretic methods without any heuristic arguments. In particular, we show that the performance of a simple algorithm that requires only a simple matrix-vector multiplication cannot be improved significantly in some practical scenarios, especially if coarse information about the support of the angular power spectrum is available.
\end{abstract}
\section{Introduction}
\label{sec:intro}

Owing to the lack of channel reciprocity in the frequency division duplex (FDD) mode, most proposals for FDD massive multiple-input multiple-output (MIMO) systems  
envision some sort of feedback from the user equipment for downlink channel state acquisition at the base stations \cite{ad2013,dec2015,miretti18,miretti18SPAWC,hag2018multi,xie2016overview,dai2018}. If traditional channel feedback mechanisms are employed, and the number of antennas is  large, downlink channel estimation in FDD massive MIMO systems may incur a prohibitive feedback overhead. Therefore, approaches for reducing this overhead have received a great deal of attention in recent years. Many promising approaches for reducing this overhead exploit information about the downlink covariance matrix \cite{ad2013,xie2016overview}, which is estimated from the uplink covariance matrix \cite{dec2015,miretti18,miretti18SPAWC,hag2018multi}. In this study, we address this estimation problem, hereafter called the uplink-downlink conversion problem.

Although channel reciprocity is lost in FDD massive MIMO systems, estimating the downlink covariance matrix from the uplink covariance matrix is possible by exploiting a different form of reciprocity, the so-called reciprocity of the angular power spectrum \cite{miretti18,miretti18SPAWC,xie2016overview,hag2018multi}. The basic assumption behind this form of reciprocity is that the average receive/transmit power at a unit of angle (hereafter called {\it angular power spectrum}) at an antenna array is frequency invariant because, if the frequency separation is not too large, the scattering environments are the same for both the uplink and downlink channels. Building upon this characteristic of wireless channels, researchers have proposed to estimate the angular power spectrum from the uplink covariance matrix, with the intent to use this estimate in models of the antenna array to recover the downlink covariance matrix \cite{miretti18,miretti18SPAWC,hag2018multi} .

In particular,  the algorithms in \cite{miretti18,miretti18SPAWC} estimate the angular power spectrum with set-theoretic methods that can easily include side information expressed in terms of closed convex sets in a Hilbert space.  Despite working in possibly infinite dimensional spaces, one of the approaches in \cite{miretti18,miretti18SPAWC} have shown that good uplink-downlink conversion performance can be obtained with a very simple matrix-vector multiplication. In that scheme with remarkably low computational complexity, the matrix is computed only once for the entire system lifetime, and the vector is constructed by rearranging the components of the uplink covariance matrix to be converted. If additional information about the angular power spectrum is used, the studies in \cite{miretti18,miretti18SPAWC} have also shown that the conversion performance of set-theoretic approaches can be further improved with simple fixed point algorithms that do not appeal to finite-dimensional approximations of the physical models. These approaches can easily take into account the polarization of antennas and real-world impairments (e.g., dissimilarities of the antennas in the array), but performance bounds on the conversion performance have not been considered in \cite{miretti18,miretti18SPAWC}.

More recently, by using ideal antenna models in uniform linear arrays, the study in \cite{hag2018multi} has proved that, in uplink-downlink channel covariance conversion based on algorithms that first estimate the angular power spectrum (such as those in \cite{miretti18,miretti18SPAWC}), some of the components of the downlink covariance matrix can be reliably reconstructed. Based on this observation, that study has derived a scheme in which the angular power spectrum is first estimated by using solvers for nonnegative least square problems (this first step can be interpreted as a finite-dimensional approximation of a particular case of \cite[Algorithm~2]{miretti18}). This estimate is then used to reconstruct the downlink covariance matrix, and, based on a formal analysis of the reliability of the reconstruction, the authors of \cite{hag2018multi} have proposed to set to zero the components of the downlink covariance matrix that are not guaranteed to be reliably estimated, in an approach called \emph{truncation}. However, setting to zero these components is a somewhat heuristic approach that can actually decrease the performance of the reconstruction in some scenarios, as the simulations in \cite{hag2018multi} have already shown. Furthermore, the reliability analysis does not seem easy to extend to realistic propagation and antenna models such as those in \cite{miretti18SPAWC}, or to cases where the antenna array response is measured to mitigate modeling errors.

Inspired by the findings in \cite{hag2018multi}, we derive novel bounds on the reconstruction error of each component of the downlink covariance matrix. Unlike previous results, the proposed performance bounds do not assume any particular antenna array or propagation model, so they can be easily applied to the realistic models in \cite{miretti18SPAWC} (without any changes) or to cases where the array response is measured. In addition, if the angular power spectrum is normalized (as assumed in \cite{hag2018multi}), a possible value for the only unknown multiplying constant in the bounds is trivial to determine. The bounds are based on elementary arguments in Hilbert spaces, and they also provide insights to improve set-theoretic algorithms. In particular, if information about the support of the angular power spectrum is used appropriately, we show that the simplest of the algorithms for uplink-downlink conversion in \cite{miretti18,miretti18SPAWC} can be so effective that performance improvements obtained with more computationally demanding conversion mechanisms are marginal at best (NOTE: the proposed bounds can also be used in the analysis of these complex mechanisms). This result is particularly appealing because that simple algorithm has no parameters to be tuned, and all steps of the enhancements we propose are justified by rigorous arguments.

This study is structured as follows. In Sect.~\ref{sect.preliminaries} we introduce the main mathematical concepts, and we prove a simple result (Proposition~\ref{prop.general_bound}) that is the main mathematical tool used to derive the novel bounds. The general result in Sect.~\ref{sect.preliminaries} is specialized to the problem of uplink-downlink covariance matrix conversion in Sect.~\ref{sect.error_bounds}, which also discusses how to exploit the proposed bounds to enhance set-theoretic methods (see Sect.~\ref{sect.side_info}). To keep the presentation as general as possible, we do not assume any particular antenna array or propagation model in Sect.~\ref{sect.error_bounds}. A concrete application of the theory developed here is shown in Sect.~\ref{sect.ULA}.

\section{Mathematical preliminaries}
\label{sect.preliminaries}
In the following, we use $\real_+$ to indicate nonnegative reals. The (coordinate-wise) real and imaginary components of complex vectors or matrices are given by, respectively, $\mathrm{Real}({\cdot})$ and $\mathrm{Imag}(\cdot)$, and $i$ is the imaginary unit, which is the solution to $i^2=-1$. By $(\cdot)^T$ we denote the transpose of a matrix or vector. Given a matrix $\signal{M}\in\real^{Q\times N}$, $\mathrm{vec}(\signal{M})\in\real^{QN}$ is the vector obtained by stacking the columns of $\signal{M}$, and $\signal{M}^\dagger$ is the Moore-Penrose (pseudo-)inverse of $\signal{M}$. We denote by $\mathcal{H}$ a real Hilbert space with inner product $\innerprod{\cdot}{\cdot}$ and norm $\|\cdot\|=\sqrt{\innerprod{\cdot}{\cdot}}$. Given ${x}\in\mathcal{H}$ and a set $C\subset\mathcal{H}$, we define $x+C:=\{h+x \in\mathcal{H}~|~h\in C\}$. A {\it linear variety} $V\subset\mathcal{H}$ is a set that can be expressed as $V=x+M$ for a vector  $x\in\mathcal{H}$ and a subspace $M\subset\mathcal{H}$; i.e., $V$ is a translation of the subspace $M$. If $C\subset \mathcal{H}$ is a nonempty closed convex set, the projection $P_C:\mathcal{H}\to~C$ maps a vector $x\in \mathcal{H}$ to the uniquely existing vector $y_0\in C$ satisfying $(\forall y\in C)~\|x-y_0\|\le\|x-y\|$.  The orthogonal complement of a subset $C\subset \mathcal{H}$ is the closed subspace given by $C^\perp:=\{y\in\mathcal{H}~|~(\forall x\in C)~\innerprod{x}{y}=0\}$, and note that $M\cup M^\perp=\mathcal{H}$ and $(M^\perp)^\perp = M$ for any closed subspace $M$. The closure of a set $C\subset\mathcal{H}$ is denoted by $\overline{C}$. A set $S=\{x_1,\ldots,x_N\}\subset\mathcal{H}$ is called {\it linear independent} (respectively, dependent) if the vectors $x_1,\ldots,x_N$ are linearly independent (respectively, dependent).  Given a function $f:\Omega\to \real$ with  $\Omega\subset\real^N$, we define its support to be the set  $\mathrm{Supp}(f)=\{x\in\Omega~|~f(x)\neq 0\}$. By $L^2(\Omega)$, $\Omega\subset \real^N$, we denote the space of real-valued square-integrable functions $f:\Omega\to\real$  with respect to the standard Lebesgue measure.

 Below is a summary of standard results in convex analysis that we use throughout this study. The proof can be found in most standard references on convex and functional analysis (e.g., \cite{yukawa2010,luen}).

\begin{fact} 
	\label{fact.basic}
	\begin{itemize}
		\item[(i)] Let $M\subset \mathcal{H}$ be a closed (linear) subspace. Then  $(\forall x\in\mathcal{H})~ x=P_M(x)+P_{M^\perp} (x)$.
		Furthermore, the projection $P_M:\mathcal{H}\to M$ onto $M$ is a bounded linear operator with operator norm given by $\|P_M\|_\mathrm{o}:=\sup_{\|x\|=1}{\|P_M(x)\|}\le 1$, and the equality is achieved if $M\neq\{0\}$.				
		\item[(ii)] Let $M\subset\mathcal{H}$ be a closed subspace. For a given $u\in\mathcal{H}$, consider the closed linear variety $V=u+M$. Then $(\forall x\in V)~V=x+M$ and $x=P_V(0)+P_M(x)$.
		\item[(iii)] Let $V=\cap_{k=1}^K\{x\in\mathcal{H}~|~\innerprod{x}{v_k}=b_k\}\neq\emptyset$, where $(v_k,~b_k)\in\mathcal{H}\times\real$ for each $k\in\{1,\ldots,K\}$. Then we have $(\forall u\in V)~V=u+{M}^\perp$, where $M=\mathrm{span}\{v_1,\ldots, v_K\}$.

	\end{itemize} 
\end{fact}

In the application described in the next section, we study the performance of algorithms producing an estimate $\widetilde{\rho}\in V$ of $\rho\in V$, where $V$ is a closed linear variety generated by the translation of the orthogonal complement $M^\perp$ of a finite  dimensional subspace $M$ (NOTE: the closed subspace $M^\perp$ can be infinite dimensional). In that application, we are not directly interested in the error $\|\widetilde{\rho}-\rho\|$, but in the approximation of $\innerprod{\rho}{y}$ by $\innerprod{\widetilde{\rho}}{y}$ for a given  $y\in\mathcal{H}$. The absolute error of the approximation is given by $e:=\left|\innerprod{\widetilde{\rho}-\rho}{y}\right|$, and we show in Proposition~\ref{prop.general_bound} elementary bounds for $e$ that decouples into the product of two terms. The important aspect to note is that one of these terms only depends on  the choice of the algorithm, but not on the estimand $\rho$. 

\begin{proposition}
\label{prop.general_bound} Let $M\subset\mathcal{H}$ be a closed subspace, and consider the linear variety $V=\rho+M^\perp$ for a given $\rho\in\mathcal{H}$. Suppose that an algorithm produces an estimate $\widetilde{\rho}\in {V}$ of $\rho\in {V}$. Then each of the following holds:
\begin{enumerate}
	\item[(i)] $(\forall y\in \mathcal{H})$ 
	\begin{multline}
		\label{eq.decoupled_bound}
		\left|\innerprod{\widetilde{\rho}-\rho}{y}\right|\le \|\widetilde{\rho} - \rho\|~ \|y-P_{M}(y)\| \\ = \|P_{M^\perp}(\widetilde{\rho}) - P_{M^\perp}(\rho)\|~ \|y-P_{M}(y)\|  
	\end{multline}
	\item[(ii)] If $\widetilde{\rho}=P_V(0)$, then 
	\begin{multline}
		\label{eq.bound_algo1}
		(\forall y\in \mathcal{H})~\left|\innerprod{\widetilde{\rho}-\rho}{y}\right|\le \|\rho - P_{M}(\rho)\|~ \|y-P_{M}(y)\| 
	\end{multline}
\end{enumerate}
\end{proposition}
\begin{proof}
	(i) Let $y\in\mathcal{H}$ be arbitrary. By assumption, both $\widetilde{\rho}$ and $\rho$ are elements of the linear variety $V$, so we have 
	\begin{multline}
	\label{eq.error_x}
	\widetilde{\rho}-\rho = P_{M^\perp}(\widetilde{\rho})+P_V(0)-P_{M^\perp}(\rho)-P_V(0) \\ = P_{M^\perp}(\widetilde{\rho})-P_{M^\perp}(\rho)= P_{M^\perp}(\widetilde{\rho}-\rho)
	\end{multline}
	 by Fact~\ref{fact.basic}(i)-(ii).  From $y=P_{M}(y)+P_{M^\perp}(y)$  (Fact~\ref{fact.basic}(i)) and the definition of orthogonal complements, we obtain 
	\begin{multline*}
	\left|\innerprod{\widetilde{\rho}-\rho}{y}\right| = \left|\innerprod{P_{M^\perp}(\widetilde{\rho}-\rho)}{P_{M}(y)+P_{M^\perp}(y)}\right| \\ = \left|\innerprod{P_{M^\perp}(\widetilde{\rho}-\rho)}{P_{M^\perp}(y)}\right|. 
	\end{multline*}
	A direct application of the Cauchy-Schwartz inequality yields $\left|\innerprod{\widetilde{\rho}-\rho}{y}\right| \le \|P_{M^\perp}(\widetilde{\rho}-\rho)\|~ \|P_{M^\perp}(y)\|$. The result in \refeq{eq.decoupled_bound} now follows from 
$P_{M^\perp}(y) = y - P_{M}(y)$ (Fact~\ref{fact.basic}(i)) and \refeq{eq.error_x}.

	(ii) By Fact~\ref{fact.basic}(i)-(ii) and $\widetilde{\rho}=P_V(0)$, we deduce $P_{M^\perp}(\widetilde{\rho}) = 0$. Now use the equality $\rho=P_{M^\perp}(\rho)+P_{M}(\rho)$ (Fact~\ref{fact.basic}(i)) in \refeq{eq.decoupled_bound} to obtain the desired result.
\end{proof}

 Proposition~\ref{prop.general_bound} has a very natural interpretation. If the estimation error $\|\widetilde{\rho}-\rho\|$ of $\rho\in V$ is bounded, then \refeq{eq.decoupled_bound} shows that the error $e:=\left|\innerprod{\widetilde{\rho}-\rho}{y}\right|$ is small if there exists a vector $u\in M$ that is sufficiently close to $y\in\mathcal{H}$ with respect to the metric $d(u,y)=\|u-y\|$. In this case, the choice of the algorithm used to produce the estimate $\widetilde{\rho}\in V$ does not play a decisive role in the minimization of the error $e$. Proposition~\ref{prop.general_bound}(ii) shows the guaranteed performance bound of a simple algorithm for estimating $\rho$. For this scheme, the error $e$ is small if $\rho$ or $y$, or both, can be well approximated by vectors in the subspace $M$.

\section{Error bounds for uplink-downlink conversion in FDD MIMO systems}
\label{sect.error_bounds}

In this section, we apply the results of Sect.~\ref{sect.preliminaries} to the problem of  covariance matrix conversion in FDD massive MIMO systems, which, as mentioned in the introduction, we call the uplink-downlink conversion problem. We first describe the problem in Sect.~\ref{sect.system_model}, and then we proceed to tailor the bounds in Proposition~\ref{prop.general_bound} to our particular application in Sect.~\ref{sect.bounds_array}. In Sect.~\ref{sect.side_info}, we show how the  bounds can be used to improve the approaches in \cite{miretti18,miretti18SPAWC}. To keep the discussion as general as possible, we do not assume any particular array geometry or propagation model in this section. 

\subsection{The uplink-downlink conversion problem}
\label{sect.system_model}
  We consider a single-cell flat-fading wireless system, in which a base station equipped with N antennas exchanges data with a single-antenna user. In the uplink, the base station first estimates the uplink channel covariance matrix $\signal{R}_\ul = E[\signal{h}_\ul \signal{h}_\ul^H] \in \mathbb{C}^{N\times N}$ from   samples of the uplink channel $\signal{h}_\ul$ and any prior knowledge of this covariance matrix. In the uplink-downlink conversion problem, samples  of the downlink channel $\signal{h}_\dl$ are not available at the base station, and the objective is to obtain an estimate of the downlink channel covariance matrix $\signal{R}_\dl = E[\signal{h}_\dl \signal{h}_\dl^H]\in \mathbb{C}^{N\times N}$ directly from the estimate of $\signal{R}_\ul$. The main challenge for the conversion in FDD MIMO systems is the lack of channel reciprocity. The uplink and downlink channels use different frequencies, so their statistics are also different, which in turn implies that $\signal{R}_\ul\neq\signal{R}_\dl$. However, $\signal{R}_\dl$ and $\signal{R}_\ul$ are related. In particular, estimating $\signal{R}_\dl$ from $\signal{R}_\ul$ is possible by using the so-called reciprocity of the angular power spectrum, which we now formally describe.
  
  For typical frequency separation gaps, the real and imaginary parts of each component of $\signal{R}_\ul$ and $\signal{R}_\dl$ can be seen as the result of an inner product in an infinite dimensional real Hilbert space $(\mathcal{H},\innerprod{\cdot}{\cdot})$, with the vectors (functions) and the inner product  taking a particular form that depends on system parameters such as the antenna polarization, array geometry, and the propagation model, among others \cite{miretti18,miretti18SPAWC}.  More precisely, let $r_{\ul, k}\in\real$ and $r_{\dl, k}\in\real$ denote one of the $k\in\mathcal{I}:=\{1,\ldots, 2N^2\}$ components of, respectively, $[\mathrm{Real}(\signal{R}_\ul)~\mathrm{Imag}(\signal{R}_\ul)]\in \real^{N\times 2N}$ and $[\mathrm{Real}(\signal{R}_\dl)~\mathrm{Imag}(\signal{R}_\dl)]\in \real^{N\times 2N}$. It has been shown in \cite{miretti18,miretti18SPAWC} that, for each $k\in\mathcal{I}$ and for a given inner product $\innerprod{\cdot}{\cdot}$ that depends on the system model,  we have 
  \begin{align}
  \label{eq.upr}
  r_{\ul, k}=\innerprod{\rho}{g_{\ul, k}}
  \end{align}
  and 
  \begin{align}
  \label{eq.dlr}
  r_{\dl, k}=\innerprod{\rho}{g_{\dl, k}},
  \end{align}
   where $\rho\in\mathcal{H}$ is the unknown frequency independent function called {\it angular power spectrum}, and $g_{\ul, k}\in\mathcal{H}$ and $g_{\dl, k}\in\mathcal{H}$ are known uplink and downlink functions related to the antenna array responses (see Sect.~\ref{sect.ULA} for a concrete example). Intuitively, the angular power spectrum is a function that shows the average angular power density that an array receives from the user at a given azimuth (and possibly elevation) angle. In the literature \cite{miretti18,miretti18SPAWC,xie2016overview,hag2018multi}, it is assumed to be the same for both the uplink and downlink channels, which is the phenomenon we call reciprocity of the angular power spectrum.
  
   With the above explanations, we can summarize the set-theoretic approaches in \cite{miretti18,miretti18SPAWC} (and some of the approaches in \cite{hag2018multi}) to the uplink-downlink conversion problem with the following two steps:
  \begin{itemize}
  	\item[(i)] We first obtain an estimate $\widetilde{\rho}\in\mathcal{H}$ of $\rho\in\mathcal{H}$ from the equations $(\forall k\in\mathcal{I})~r_{\ul, k}=\innerprod{\rho}{g_{\ul, k}}$ and possibly known properties of $\rho\in\mathcal{H}$ by using set-theoretic methods.
  	\item[(ii)] With $\widetilde{\rho}\in\mathcal{H}$, we obtain an estimate $\widetilde{r}_{\dl, k}\in\real$ of ${r}_{\dl, k}\in\real$ for each $k\in\mathcal{I}$ by computing $\widetilde{r}_{\dl, k}=\innerprod{\widetilde{\rho}}{g_{\dl, k}}$.
  \end{itemize}
 One of the contributions of the next section is to derive conditions guaranteeing that, given $k\in\mathcal{I}$, the estimate $\widetilde{r}_{\dl,k}$ of ${r}_{\dl,k}$ in step (ii) is accurate even if the estimate $\widetilde{\rho}$ of $\rho$ in step (i) is inaccurate. These conditions will be used to derive simple techniques to improve set-theoretic methods addressing the problem in step (i).

 \subsection{Bounds for the error of UL-DL covariance conversion with general arrays}
 \label{sect.bounds_array}
  We now derive performance bounds for the set-theoretic approaches described above. To this end, we assume that the uplink covariance matrix $\signal{R}_\ul$ (and hence $(r_{\ul, k})_{k\in\mathcal{I}}$) is perfectly estimated. By recalling that covariance matrices have structure (they are at least Hermitian), the number of different equations in \refeq{eq.upr} and \refeq{eq.dlr} is strictly less than $|\mathcal{I}|=2N^2$ ($|\mathcal{I}|$ denotes the cardinality of $\mathcal{I}$). Therefore, many repeating equations can be removed, but for brevity this simple operation is not considered in this section. 

To proceed with the bounds, we define  
\begin{align}
\label{eq.setSp}
S^\prime=\{g_{\ul,1},\ldots, g_{\ul,|\mathcal{I}|}\}\subset\mathcal{H}
\end{align}
to be the set corresponding to the uplink functions in \refeq{eq.upr}. The angular power spectrum $\rho\in\mathcal{H}$ is related to $S^\prime$ by the fact that, from \refeq{eq.upr} and Fact.~\ref{fact.basic}(ii)-(iii), we have  $\rho\in V^\prime:=\cap_{k\in\mathcal{I}}\{x \in\mathcal{H}~|~\innerprod{x}{g_{\ul, k}} = r_{\ul, k}\}=\rho+\mathrm{span}(S^\prime)^\perp$. Intuitively, the linear variety $V^\prime\subset\mathcal{H}$ is the set containing all angular power spectrum functions that produce the same uplink covariance matrix.

To include in the analysis any prior information about $\rho\in\mathcal{H}$ expressed in terms of closed linear varieties or closed subspaces, we assume that
\begin{align}
\label{eq.Vdprime}
\rho\in V^\dprime:=\cap_{k=1}^Q\{x\in\mathcal{H}~|~\innerprod{x}{v_k}=b_k\} 
\end{align}
and that the tuples $\{(v_k, b_k)\}_{k=1,\ldots,Q}\subset \mathcal{H}\times \real$ used to construct the linear variety $V^\dprime$ ($V^\dprime$ is a subspace if $b_1=\ldots=b_Q=0$) are  known. We now define $V:=V^\prime\cap V^\dprime$ and construct a new set $S\subset\mathcal{H}$ containing all vectors in $S^\prime$ in \refeq{eq.setSp} and all vectors in $S^\dprime:= \{v_{1},\ldots, v_{Q}\}$; i.e.,  
\begin{align}
\label{eq.setS}
S:=S^\prime\cup S^\dprime \subset\mathcal{H}.
\end{align}

By recalling that a nonempty intersection of closed linear varieties is a closed linear variety, the set $V=V^\prime \cap V^\dprime\ni\rho$ defined above is a closed linear variety $V$ that can be equivalently written as $V=\rho+M^\perp$, where $M\subset\mathcal{H}$ is the closed subspace $M:=\mathrm{span}(S)$. With these operations, we are now exactly in the setting of Proposition~\ref{prop.general_bound}. Before we proceed with the specialization of this proposition to the problem of uplink-downlink conversion, we first show that the projections $P_M:\mathcal{H}\to M$ and $P_V(0)\in\mathcal{H}$ are easy to compute.

To simplify the notation, denote by $x_1, \ldots, x_{L}$ the $L=|\mathcal{I}|+Q$ vectors in the set $S$ in \refeq{eq.setS}. Without any loss of generality, we assume that $x_k=g_{\ul,k}$ for $k\in\{1,\ldots,|\mathcal{I}|\}$ and $x_{k}=v_{k-|\mathcal{I}|}$ for $k\in\{|\mathcal{I}|+1,\ldots,|\mathcal{I}|+Q\}$. Define the following matrix:

\begin{align}
\label{eq.matrixg}
\signal{G} = \left[\begin{matrix}\innerprod{x_1}{x_1}&\cdots&\innerprod{x_1}{x_{L}}\\
\vdots&\ddots&\vdots \\
\innerprod{x_{L}}{x_1} & \cdots &  \innerprod{x_{L}}{x_{L}}
\end{matrix}\right] \in\real^{L \times L}.
\end{align}

With the above definitions, we can use arguments similar to those in \cite[Ch. 3]{luen}\footnote{Here we do not assume $S$ to be a linearly independent set.} to show that the projection from $y\in\mathcal{H}$ onto the closed subspace $M$ is given by:
\begin{align}
\label{eq.projm}
P_M:\mathcal{H}\to M:y\mapsto \sum_{k=1}^{L}\alpha_k x_{k},
\end{align}
where $\signal{\alpha}=[\alpha_1,\ldots,\alpha_{L}]^T \in\real^{L}$ is any solution to $\signal{G}\signal{\alpha}=\signal{z}$, and $\signal{z}=[\innerprod{x_1}{y} \ldots \innerprod{x_{L}}{y}]^T\in\real^{L}.$ In turn, the projection from $0$ onto the linear variety $V$  described above is given by 
\begin{align}
\label{eq.pv0}
P_V(0)=\sum_{k=1}^L\beta_k x_k,
\end{align}
where $\signal{\beta}=[\beta_1,\ldots,\beta_L]^T\in\real^L$ is any solution to $\signal{G}\signal{\beta}=[r_{\ul, 1},\ldots,r_{\ul, |\mathcal{I}|},b_1,\ldots,b_Q]^T\in\real^L$.

As it will soon become clear, in the proposed performance bounds we are specially interested in the estimation error $\|g_{\dl, k}-P_M(g_{\dl, k})\|$ for each $k\in\mathcal{I}$, which can be easily computed as shown in the following standard result. We omit the proof for brevity, but it can be easily obtained by using \refeq{eq.projm}. 

\begin{proposition}
	\label{proposition.error_qdk} Let $P_M(y)\in M=\mathrm{span}(\{x_1,\ldots,x_{L}\})$ be the approximation in the subspace $M$ of an arbitrary vector $y\in\mathcal{H}$. Then the approximation error $\|y-P_M(y)\|$ is given by 
	\begin{align*}
	 \|y-P_M(y)\|=\sqrt{(\|y\|^2 - \signal{z}^T \signal{G}^{\dagger} \signal{z})},
	 \end{align*}
	  where  $\signal{z}=\signal{G}\signal{\alpha}\in\real^L$, $\signal{\alpha}\in\real^L$, and $\signal{G}\in\real^{L\times L}$ are as defined above.
\end{proposition}

Now, let
\begin{align}
\label{eq.qmatrix}
\signal{Q} := \left[\begin{matrix}\innerprod{x_1}{g_{\dl,1}}&\cdots&\innerprod{x_1}{g_{\dl,|\mathcal{I}|}}\\
\vdots&\ddots&\vdots \\
\innerprod{x_{L}}{g_{\dl,1}} & \cdots &  \innerprod{x_{L}}{g_{\dl,|\mathcal{I}|}}
\end{matrix}\right] \in\real^{L\times |\mathcal{I}|}.
\end{align}

The proposed error bounds for uplink-downlink conversion are shown in the next corollary. 

\begin{Cor}
	\label{cor.general_bound} Denote by $\signal{q}_k\in\real^{L}$ the $k$th column of the matrix $\signal{Q}$ in \refeq{eq.qmatrix}, and let $\signal{G}$ be as defined in \refeq{eq.matrixg}. Suppose that $\tilde{\rho}\in V$ is an estimate of the angular power spectrum $\rho\in V$ obtained by a given algorithm, where $V=\rho+M^\perp$ and $M\subset\mathcal{H}$ is the closed subspace $M=\mathrm{span}(S)$ with $S$ as defined in \refeq{eq.setS}. Further, assume that $\|\rho\|\le B$ for some $B\in\real$. Let $\innerprod{\widetilde{\rho}}{g_{\dl, k}}=\widetilde{r}_{\dl, k}$ be the estimate of the $k$th ($k\in\mathcal{I}$) component $r_{\dl, k}$ of the downlink covariance matrix $\signal{R}_\dl$. Then the estimation error $e_k := |\widetilde{r}_{\dl, k}-r_{\dl, k}|$ for each $k\in\mathcal{I}$ satisfies the following: 
	\begin{itemize}
		\item[(i)] If the algorithm used to produce the estimate $\widetilde{\rho}\in V$ also guarantees $\|\widetilde{\rho}\|\le B$, then 
		\begin{multline*}
			(\forall k\in\mathcal{I})~ e_k \le \|\rho-\widetilde{\rho}\|~ \|g_{\dl,k}-P_M(g_{\dl,k})\| \\ \le  2B \sqrt{(\|g_{\dl, k}\|^2 - \signal{q}_k^T \signal{G}^{\dagger} \signal{q}_k)}.
		\end{multline*} 
		
		\item[(ii)] Using $\widetilde{\rho}=P_V(0)$ as the estimate of the angular power spectrum $\rho$, we have $(\forall k\in\mathcal{I})$
		\begin{multline}
		\label{eq.bound_cor}
		e_k \stackrel{(a)}{\le} \|\rho-P_M(\rho)\| ~ \|g_{\dl,k}-P_M(g_{\dl,k})\|
		\\ \stackrel{(b)}{\le} \|\rho\|~\|g_{\dl,k}-P_M(g_{\dl,k})\|\\ \stackrel{(c)}{\le} B \sqrt{(\|g_{\dl, k}\|^2 - \signal{q}_k^T \signal{G}^{\dagger} \signal{q}_k)}.
		\end{multline} 
	\end{itemize}
\end{Cor}
\begin{proof} 
	
The proof of (i) is immediate from Proposition~\ref{prop.general_bound}(i), Proposition~\ref{proposition.error_qdk}, and the triangle inequality. To prove (ii), we note that the inequality in $(a)$ follows from Proposition~\ref{prop.general_bound}(ii), the inequality in $(b)$ follows from $\|\rho-P_M(\rho)\|=\|P_{M^\perp}(\rho)\|\le \|P_{M^\perp}\|_\mathrm{o}~\|\rho\| \le  \|\rho\|$ [see Fact~\ref{fact.basic}(i)], and the inequality in $(c)$ follows from  Proposition~\ref{proposition.error_qdk} and the assumption $\|\rho\|\le B$.
\end{proof}

\subsection{Improving the performance of the conversion with information about the support of the angular power spectrum}
\label{sect.side_info}

One of the practical implications of Corollary~\ref{cor.general_bound} is that, for a given  $k\in\mathcal{I}$, any algorithm producing an estimate $\widetilde{\rho}\in V$ of $\rho\in V$ is able to approximate reliably the components $r_{\dl, k}$ of the downlink covariance matrix $\signal{R}_{\dl, k}$ provided that the term $\|g_{\mathrm{d},k}-P_M(g_{\dl,k})\|$ is sufficiently small, regardless of how challenging the scenario for the estimation of $\rho$ may be. By recalling that the projection $P_M(g_{\dl,k})\in\mathcal{H}$ can be interpreted as the best approximation of $g_{\dl,k}$ in the closed subspace $M$, adding to the subspace $M$ functions as similar as possible to $g_{\dl,k}$  is a natural idea to decrease the estimation error bound $2B\|g_{\mathrm{d},k}-P_M(g_{\dl,k})\|$ of $r_{\dl, k}$. In the discussion below, we show a simple technique to design $M$ based on this simple principle.

The subspace $M$ is by definition the span of $S=S^\prime\cup S^\dprime$, where $S^\prime=\{g_{\ul,1},\ldots,g_{\ul,|\mathcal{I}|}\}$ is the set of uplink functions and $S^\dprime=\{v_1,\ldots,v_Q\}$ is the set of functions resulting from any prior knowledge about $\rho$ (see \refeq{eq.Vdprime}). Therefore, a simple means of including in the subspace $M$ functions that are close to each of the downlink functions $(g_{\dl,k})_{k\in\mathcal{I}}$ is to make the uplink functions $(g_{\ul,k})_{k\in\mathcal{I}}$ as similar as possible to the downlink functions $(g_{\dl,k})_{k\in\mathcal{I}}$. Alternatively, we can also include in the set $S^\dprime$ functions that are as similar as possible to the downlink functions $(g_{\dl,k})_{k\in\mathcal{I}}$ (NOTE: including $(g_{\dl,k})_{k\in\mathcal{I}}$ directly while guaranteeing $\rho\in V$ is difficult).  The first approach, which corresponds to the design of uplink and downlink functions, may not be always possible because it typically entails changes in hardware (e.g., changes in the inter-antenna spacing) or other modifications in standardized system parameters (e.g., operating frequencies). Therefore, here we focus on the second approach; namely, the construction of an appropriate set $S^\dprime$, or, equivalently, the corresponding linear variety or subspace $V^\dprime$ in \refeq{eq.Vdprime}. To derive the sets, we further assume the following:

\begin{itemize}
	\item[(A1)] The angular power spectrum $\rho\in\mathcal{H}$, the downlink functions $(g_{\dl,k})_{k\in \mathcal{I}}\subset\mathcal{H}$, and the uplink functions $(g_{\ul,k})_{k\in \mathcal{I}}\subset\mathcal{H}$ are functions in a Hilbert space of functions in $L^2(\Omega)$, or, as in \cite{miretti18SPAWC}, a Hilbert space $\mathcal{H}$ of tuples in $L^2(\Omega)\times L^2(\Omega)$ (NOTE: extensions to different Hilbert spaces is straightforward),~\footnote{In these Hilbert spaces, which are used in Sect.~\ref{sect.ULA}, we typically work with classes of equivalent functions, with the equivalence relation between two functions $f$ and $g$ defined by $f\sim g\Leftrightarrow \|f-g\|=0$. Equalities such as $f=g$ should be understood as equalities between the classes, not to the particular functions (in a pointwise sense) because $f$ and $g$ can differ, for example, in a countable set in their domains.}  where $\Omega\subset\real^K$. \\
	\item[(A2)] There exists a known non-null measurable set $C_\mathrm{S}\subset \Omega$ such that $\mathrm{Supp}(\rho)\subset C_\mathrm{S}\neq \emptyset$ for the angular power spectrum functions $\rho\in\mathcal{H}$ that can be observed in the system. Intuitively, the set $C_\mathrm{S}$ is a superset of $\mathrm{Supp}(\rho)$ for which $\theta\notin C_\mathrm{S}$ implies $\rho(\theta)=0$.
\end{itemize}
Assumption A1 is very natural. It is satisfied in many realistic models representing the angular power spectrum in practical systems. In these models, $\Omega$ is the set of possible azimuth and elevation angles \cite{miretti18,miretti18SPAWC}. Assumption A2 is system dependent, but it may be valid in scenarios where signals of users impinging on the antenna array are not likely to have any significant power at certain angles, which are used for the construction of $C_\mathrm{S}$. 

We now proceed to show how support information of $\rho$ can be used to design the subspace $M$ by using arguments that have a strong theoretical justification. To this end, consider the closed subspace
 \begin{align*}
\mathcal{K}:=\overline{{\{x\in\mathcal{H}~|~(\forall\theta\in  C_\mathrm{S})~ x(\theta)=0\}}}.
\end{align*}
 The projection $P_\mathcal{K}:\mathcal{H}\to\mathcal{K}$ from $v\in\mathcal{H}$ onto $\mathcal{K}$ is the function given by (we omit the proof for brevity):
\begin{align*}
\mathcal{H}\ni P_{\mathcal{K}}(v):\Omega\to\real:\theta\mapsto\begin{cases}
0,&\text{if }\theta\in C_{\mathrm{S}}, \\
 v(\theta) & \text{otherwise}.
\end{cases}
\end{align*}
Since $P_{\mathcal{K}}(\rho) = 0$ from the assumption $\mathrm{Supp}(\rho) \subset C_\mathrm{S}$ and the definition of the subspace $\mathcal{K}$, we have $\rho\in \mathcal{K}^\perp$, and thus
\begin{align} 
\label{eq.support_info}
(\forall v\in\mathcal{H})\innerprod{P_\mathcal{K}(v)}{\rho}=0.
\end{align}
  In particular, using the downlink functions as the function $v$ in \refeq{eq.support_info} yields
\begin{align}
\label{eq.projgk}
(\forall k\in\mathcal{I})\innerprod{P_\mathcal{K}(g_{\mathrm{d},k})}{\rho}=0.
\end{align}
We have now reached the point to show the closed subspace $V^\dprime$ we propose to represent the prior knowledge about the support of $\rho$. More precisely, in light of \refeq{eq.projgk}, we use 
\begin{align*}
V^\dprime=\cap_{k\in\mathcal{I}}\{x\in\mathcal{H}~|~\innerprod{x}{v_k}=0\}\ni \rho,
\end{align*}
 where $v_k:=P_\mathcal{K}(g_{\mathrm{d},k})$ for each $k\in \mathcal{I}$. This choice is intuitively appealing because we add to the set $S$ in \refeq{eq.setS} all vectors $(v_k)_{k\in{\mathcal{I}}}$ in $\mathcal{K}$ that best approximate (with respect to the metric $d(x,y)=\|x-y\|$) the downlink functions $(g_{\dl,k})_{k\in{\mathcal{I}}}$, and we recall from the above discussion that, for each $k\in \mathcal{I}$, the estimation error $|\widetilde{r}_{\dl,k}-r_{\dl,k}|$ decreases as the ability to represent $g_{\dl,k}$ with functions in $M=\mathrm{span}(S)$ improves. Note that we could further improve the reliability of the conversion by repeating the above procedure to include in $S$ additional functions of the form $P_\mathcal{K}(v)$ with $v\in\mathcal{H}$ [e.g., the functions $(P_\mathcal{K}(g_{\ul,k}))_{k\in\mathcal{I}}$]. Alternatively, we could also change the definition of inner products to consider only functions in $L^2(C_\mathrm{S})$. These approaches can be numerically unstable in large antenna arrays if the information about the support of $\rho$ is erroneous and appropriate mitigation techniques are not applied, but we leave this discussion to a future study because of the space limitation.

All the above improvements are available for the simple approach using $\widetilde{\rho}=P_V(0)$ as the estimate of $\rho$. This approach is particularly interesting because, as shown in the study in \cite{miretti18}, which has not considered the enhancements discussed above, the whole process of estimating the angular power spectrum and using this estimate to reconstruct the downlink covariance matrix can be done with a simple matrix-vector multiplication. This important feature is not lost with the enhancements proposed in this subsection. More precisely, denote by $\widetilde{\signal{R}}_\dl$ the estimate of the downlink covariance matrix $\signal{R}_\dl$, and recall that $\widetilde{r}_{\dl,1} = \innerprod{P_V(0)}{g_{\dl,1}},\ldots,\widetilde{r}_{\dl,|\mathcal{I}|}=\innerprod{P_V(0)}{g_{\dl,|\mathcal{I}|}}$ represent the real and imaginary parts of the components of $\widetilde{\signal{R}}_\dl$. Using \refeq{eq.pv0} in these inner products, we verify that uplink-downlink channel covariance conversion can be performed with the following simple linear operation:

\begin{align}
\label{eq.algorithm1}
\mathrm{vec}[\mathrm{Real}(\widetilde{\signal{R}}_\dl)~\mathrm{Imag}(\widetilde{\signal{R}}_\dl)] = \signal{Q}^T\signal{G}^{\dagger}\underline{\signal{r}}=\signal{A}{\signal{r}},
\end{align}
where $\underline{\signal{r}} = [{\signal{r}}^T, 0,\ldots, 0]^T\in\real^{L}$, ${\signal{r}}:=[r_{\ul,1},\ldots,r_{\ul,|\mathcal{I}|}]^T\in\real^{|\mathcal{I}|}$, and $\signal{A}\in\real^{|\mathcal{I}|\times |\mathcal{I}|}$ is the matrix obtained by keeping only the first $|\mathcal{I}|$ columns of the matrix $\signal{Q}^T\signal{G}^{\dagger}$. Note that the matrix $\signal{A}$ only depends on the support information, which is often assumed to be slowly time-varying information, and the array response, so this matrix is computed sporadically. If no support information is used, then $\signal{A}$ needs to be computed only once. 

Before we finish this section, it is also worth noticing that, by increasing the subspace $M$ with support information about $\rho$ as described above, we also decrease the algorithm error term $\|\rho-P_M(\rho)\|$ in the bound (a) in \refeq{eq.bound_cor}, thus further improving the reliability of the conversion.

 \section{Example: Uniform linear arrays}
 \label{sect.ULA}
 
We now further specialize the results in the previous section to uniform linear arrays. This particular choice enables us to relate the analysis in the previous sections to existing results in the literature that, unlike our approaches, do not seem easy to extend to schemes exploiting information about the structure of the angular power spectrum or to systems where the functions $(g_{\ul,k})_{k\in\mathcal{I}}$ and $(g_{\dl,k})_{k\in\mathcal{I}}$ are determined by measurements instead of models.\footnote{If we normalize the angular power spectrum to satisfy $\|\rho\|=1$, we can simply set the constant $B$ in Corollary~\ref{cor.general_bound} to $B=1$. In contrast, in addition to a similar  normalization, the bound in \cite[Theorem~1]{hag2018multi} requires knowledge of a constant that does not seem easy to determine. Without this constant, comparing directly the proposed bounds with the bound in \cite[Theorem~1]{hag2018multi} seems difficult.}
\vspace{-.2cm}
\subsection{System model and bounds without support information} 
 	In a uniform linear array with $N$ antennas, under very mild assumptions \cite{haghighatshoar2017massive,xie2016overview}, the uplink and downlink channel covariance matrices for typical frequency gaps are given by $\signal{R}_\ul:=\signal{R}(f_\ul)\in\mathbb{C}^{N\times N}$ and $\signal{R}_\dl:=\signal{R}(f_\dl)\in\mathbb{C}^{N\times N}$, where $f_\ul\in\real_+$ and $f_\dl\in\real_+$ are, respectively, the uplink and downlink frequencies; 
 	\begin{align*}
 	\signal{R}(f) = \int_{-\pi/2}^{\pi/2} \rho(\theta) \signal{a}(\theta, f)\signal{a}(\theta, f)^H\mathrm{d}\theta
 	\end{align*}
 	(the integral should be understood coordinate-wise) is the channel covariance matrix for a given frequency $f$;	${\rho:[-\pi/2, \pi/2]\to\real_+}$ is the angular power spectrum; 
 	\begin{equation}
 	\label{eq.integral}
 	\resizebox{\hsize}{!}{$
 	\begin{array}{rl}
 	\signal{a}:[-\frac{\pi}{2}, \frac{\pi}{2}]\times\real_+\to&\mathbb{C}^N \\
 	(\theta, f)\mapsto&\left[1,e^{i2\pi \frac{f}{c}d\sin\theta},\ldots,e^{i2\pi \frac{f}{c}d (N-1)\sin\theta}\right]
 	\end{array}$}
 	\end{equation}
 	is the array response for a given angle $\theta$ and frequency $f$; $c$ is the speed of the wave propagation; and $d$ is the inter-antenna spacing. 
 	
 	In real physical systems, we can safely assume that $\rho$ is an element of the Hilbert space $(\mathcal{H},\innerprod{\cdot}{\cdot})$ of Lebesgue (real) square-integrable functions $\mathcal{H}=L^2([-\pi/2,\pi/2])$ equipped with the inner product $(\forall \rho \in\mathcal{H})(\forall g\in\mathcal{H})\innerprod{\rho}{g}=\int_{-\pi/2}^{\pi/2}\rho(\theta)g(\theta)\mathrm{d}\theta$. As a result, by fixing $f_\ul$, in light of \refeq{eq.integral} the functions $(g_{\ul,k})_{k\in \{1,\ldots,2N^2\}}$ in \refeq{eq.upr} are obtained from the equality $(\forall\theta\in[-\pi/2,\pi/2])$
 	\begin{multline}
 	\label{eq.ordering}
 	 [g_{\ul,1}(\theta),\ldots,g_{\ul,2N^2}(\theta)]^T = \\ \mathrm{vec}\left(\left[\begin{matrix}\mathrm{Real}(\signal{a}(\theta, f_\ul)\signal{a}(\theta, f_\ul)^H)\\ \mathrm{Imag}(\signal{a}(\theta, f_\ul)\signal{a}(\theta, f_\ul)^H)\end{matrix}\right]\right).
 	\end{multline}
 	The downlink functions $(g_{\dl,k})_{k\in\{1,\ldots,2N^2\}}$ in \refeq{eq.dlr} are obtained analogously by considering the downlink frequency $f_\dl$ in \refeq{eq.ordering}.
 	
 	 In uniform linear arrays, the covariance matrices are Hermitian and Toeplitz \cite{haghighatshoar2017massive,miretti18,hag2018multi}, so, with the ordering in \refeq{eq.ordering}, we can consider only the functions $g_{\ul,1},\ldots,g_{\ul,2N}$ responsible for the first column of $\signal{R}_\ul$ because knowledge of this column is enough to reconstruct all elements of $\signal{R}_\ul$. For the same reason, we use only the downlink functions $g_{\dl,1},\ldots,g_{\dl,2N}$. By doing so, the set $S^\prime$ in \refeq{eq.setSp} is given by $S^\prime=\{g_{\ul,1},\ldots, g_{\ul,2N}\}\subset\mathcal{H}$, and we can redefine the index set $\mathcal{I}$ accordingly; i.e., $\mathcal{I}:=\{1,\ldots, 2N\}$. 
 	  
 	 Without any information about the support of $\rho$, we have $S=S^\prime$, and we can use the results in \cite[Sect.~4.1]{miretti18} to compute the algorithm independent term  $\|g_{\dl,k}-P_M(g_{\dl,k})\|$ ($k\in\mathcal{I}$) of the bounds in Corollary~\ref{cor.general_bound} by using Bessel functions of the first kind, order zero, which we denote by $J_0:\real\to\real$. In this case, the bound in the last inequality in Corollary~\ref{cor.general_bound}(ii) reduces to
 	\begin{align}
 	 	 \label{bound.specific}
 (\forall k\in \mathcal{I})~ 	 e_k\le B
\sqrt{(\|g_{\dl,k}\|^2 - \signal{q}_k^T \signal{G}^{\dagger} \signal{q}_k)},
 	\end{align}
 	where 
 	\begin{multline*}
	 	\|g_{\dl,k}\|^2=\\ \begin{cases}
	 	\dfrac{\pi}{2} \left(1+J_0\left(4\pi \dfrac{f_\dl}{c} d (k-1) \right)\right)&\text{ if } 1\le k \le N \\
		\dfrac{\pi}{2} \left(1-J_0\left(4\pi \dfrac{f_\dl}{c} d (k-N-1) \right)\right)&\text{ otherwise, } 
	 	\end{cases}
	\end{multline*}
\begin{align*}
\signal{G}=\dfrac{\pi}{2}\left[\begin{matrix}
\signal{G}_\mathrm{r}& \signal{0} \\ 
\signal{0}& \signal{G}_\mathrm{j}
\end{matrix}\right],  \signal{Q}=[\signal{q}_1,\ldots,\signal{q}_{2N}]=\dfrac{\pi}{2}\left[\begin{matrix}
\signal{Q}_\mathrm{r}& \signal{0} \\ 
\signal{0}& \signal{Q}_\mathrm{j},
\end{matrix}\right]
\end{align*}
and the components of the $n$th row and $m$th column of the matrices $\signal{G}_\mathrm{r},\signal{G}_\mathrm{j},\signal{Q}_\mathrm{r},\signal{Q}_\mathrm{j}\in\real^{N\times N}$ are given by $\signal{G}_{\mathrm{r},nm}=J_0(x_{nm})+J_0(y_{nm})$, $\signal{G}_{\mathrm{j},nm}=J_0(x_{nm})-J_0(y_{nm})$, $\signal{Q}_{\mathrm{r},nm}=J_0(p_{nm})+J_0(q_{nm})$,
$\signal{Q}_{\mathrm{j},nm}=J_0(p_{nm})-J_0(q_{nm})$ with $$x_{nm}=2\pi~d~\dfrac{f_\ul}{c}(n-m),~y_{nm}=2\pi~d~\dfrac{f_\ul}{c}(n+m-2),$$ $$p_{nm}=2\pi d\left(\dfrac{f_\ul(n-1)}{c}-\dfrac{f_\dl(m-1)}{c}\right),$$
and $$q_{nm}=2\pi d\left(\dfrac{f_\ul(n-1)}{c}+\dfrac{f_\dl(m-1)}{c}\right).$$

\subsection{Numerical experiments}
For a concrete example of the bounds, we use an antenna array with the configuration in Table~\ref{table.parameters}. The number of antennas in the array is relatively small to emphasize the fact that the proposed bounds do not appeal to asymptotic results. As discussed in \cite{hag2018multi}, the configuration in Table~\ref{table.parameters} is challenging for uplink-downlink conversion for two main reasons: (i) the uplink frequency is lower than the downlink frequency, and (ii) the antenna spacing is larger than half of the wavelength $c/(2f_\dl)$ of the higher frequency $f_\dl$, so we have the undesirable phenomenon known as grating lobes \cite{van2004optimum,hag2018multi}. We show below that this challenging scenario for uplink-downlink conversion can be formally verified with the simple bounds in \refeq{bound.specific}, and the problems for uplink-downlink conversion can be mitigated with information about the support of the angular power spectrum.

\begin{table}
	\caption{Parameters of the uniform linear array}
	\label{table.parameters}
	\begin{center}
		\begin{tabular}{cc}
			\hline \\
			Number of antennas $(N)$ & 30 \\
			Uplink frequency $(f_\ul)$ & 1.8~MHz \\
			Downlink frequency $(f_\dl)$ & 1.9~MHz \\
			Speed of wave propagation $(c)$ & $3\cdot 10^8$~m/s \\
			Antenna spacing $(d)$ & $1.05~\dfrac{c}{2f_\ul}$ \\
			\hline
		\end{tabular}
	\end{center}
	\vspace{-.5cm}

\end{table}

To illustrate the theoretical gains that can be achieved with the technique discussed in Sect.~\ref{sect.side_info}, we assume that $\mathrm{Supp}(\rho)\subset C_\mathrm{S}=[0,~\pi/2]$, and $C_\mathrm{S}$ is known. For all simulations in this section, we use only the scheme in Corollary~\ref{cor.general_bound}(ii) for the estimation of $\rho$ because of its low computational complexity, as discussed in Sect.~\ref{sect.side_info}. 

In Fig.~\ref{fig.bounds_theory}, assuming $\|\rho\|=B=1$, we show the bounds in the last inequality in \refeq{eq.bound_cor} with and without support information (SI). For the computation of the former bound, we use the expressions in \refeq{bound.specific}. For the latter bound, we construct the matrices $\signal{G}$ and $\signal{Q}$ in \refeq{eq.matrixg} and \refeq{eq.qmatrix} by computing integrals numerically, unless the integral falls into one of the cases computed  in \refeq{bound.specific}. From Fig.~\ref{fig.bounds_theory}, it is clear that, without any support information, the estimate $\widetilde{r}_{\dl,k}$  can be unreliable for many indices $k\in\mathcal{I}$, which is also in accordance with the results in \cite{hag2018multi}. In contrast, with support information, all estimates $(\widetilde{r}_{\dl,k})_{k\in\mathcal{I}}$ of the components of the downlink covariance matrix are reliable, even if the estimate $\widetilde{\rho}=P_V(0)$ of the angular power spectrum $\rho$ is not necessarily accurate. 

\begin{figure}
	\begin{center}
		\includegraphics[width=\columnwidth]{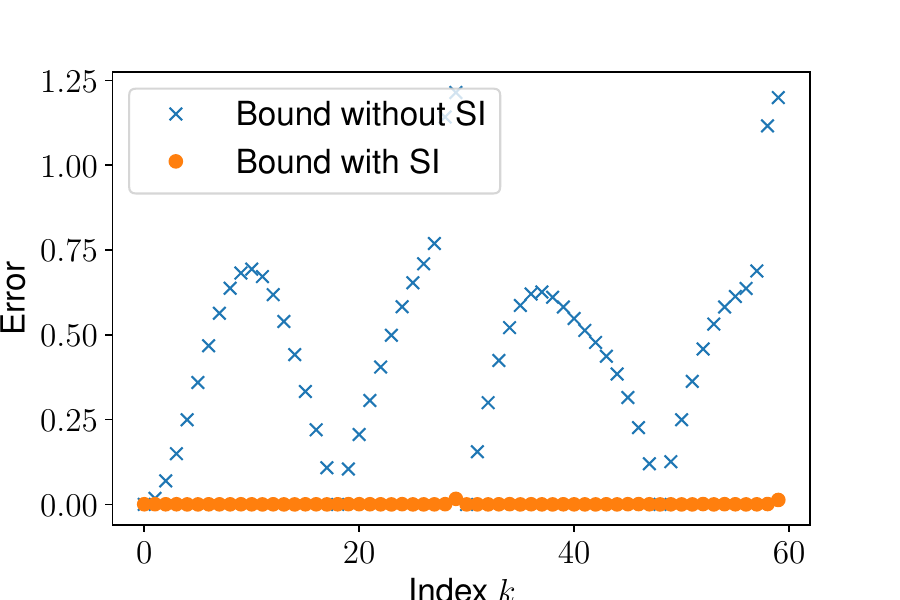}
		\caption{Upper bound \refeq{eq.bound_cor} on the error $(e_k)_{k\in\mathcal{I}}$ for the conversion performed with the algorithm in \refeq{eq.algorithm1} with and without support information. $B=1$.}
		\label{fig.bounds_theory}
	\end{center}
\end{figure}

To illustrate the above fact, consider the following example for $\rho:[-\pi/2,\pi/2]\to\real_+$:

\begin{align} 
\label{eq.aps}
\rho(\theta)= n{e}^{-{|\theta-.5|}/{.05}} + 4n~{e}^{-{|\theta-1.4|}/{.05}},
\end{align}
where $n\in\real_+$ is a normalizing constant chosen to guarantee that $\|\rho\|=1$. This exemplary $\rho$ can be interpreted as coming from a user with two multipath components at angles 0.5~rad and 1.4~rad. Note that we have violated the assumption of the support of $\rho$, but the signal energy outside $C_\mathrm{S}$ is small compared to the energy in $C_\mathrm{S}$, so we can expect the bounds shown in Fig.~\ref{fig.bounds_theory} to be accurate. This fact is illustrated in Fig.~\ref{fig.bounds_practice}, which shows the absolute error $(e_k=|r_{\dl,k}-\widetilde{r}_{\dl,k}|)_{k\in\mathcal{I}}$ [with $\tilde{r}_{\ul,k}$ computed by using \refeq{eq.algorithm1}] of the estimates. Note that, by including support information, uplink-downlink conversion has been performed reliably for all components of the downlink covariance matrix, even though the estimate of angular power spectrum (APS) is not necessarily accurate, as depicted in Fig.~\ref{fig.aps}. We verify, for example, that the schemes used to estimate the function in \refeq{eq.aps} produce functions taking negative values for some angles.

\begin{figure}
	\begin{center}
		\includegraphics[width=\columnwidth]{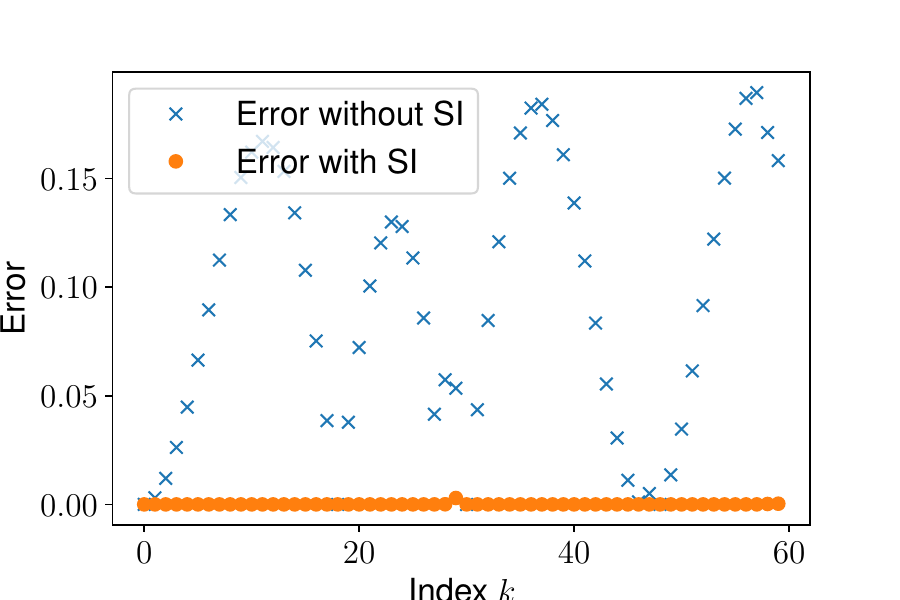}
		\caption{Simulated conversion error $(e_k)_{k\in\mathcal{I}}$ of the algorithm in \refeq{eq.algorithm1} with and without support information. Angular power spectrum in \refeq{eq.aps}. }
		\label{fig.bounds_practice}

	\end{center}
\end{figure}

\begin{figure}
	\begin{center}
		\includegraphics[width=\columnwidth]{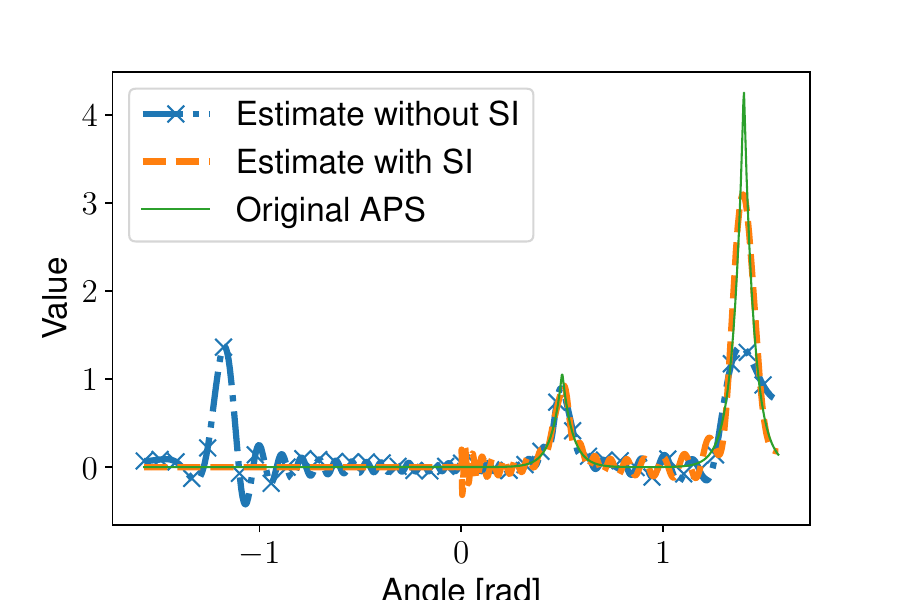}
		\caption{Estimate $\widetilde{\rho}=P_V(0)$ [see \refeq{eq.pv0}] of the angular power spectrum (APS) $\rho$ in \refeq{eq.aps} with and without support information for the construction of the linear variety $V$.} 
		\label{fig.aps}
	\end{center}
\end{figure}

\section{Summary and Conclusions}
Recent work has proved that, without side information about the angular power spectrum, existing algorithms in the literature may not be able to estimate reliably all components of the downlink covariance matrix. In this study we have introduced alternative reliability bounds that are based on elementary arguments in infinite dimensional Hilbert spaces. The main advantages of the proposed analysis are the simplicity and the generality. Unlike previous results, the bounds shown here can be straightforwardly used to analyze the performance of algorithms that exploit information about the support of the angular power spectrum in challenging scenarios that take into account the polarization of antennas and physical impairments of real antenna arrays. To illustrate a possible application of the bounds, we have improved a simple set-theoretic algorithm that does not require any parameter tuning. We have shown that, with coarse information about the angular power spectrum, all components of the downlink covariance matrix can be reliably estimated from the uplink covariance matrix with a simple linear operation. This result suggests that, in some scenarios, the main challenge may be the estimation of the uplink covariance matrix, not necessarily the uplink-downlink conversion problem. 

{\footnotesize{\bf Acknowledgment:} The work was supported by the German Federal Ministry of Education and Research under grant 16KIS0605.}


\bibliographystyle{IEEEtran}
\bibliography{IEEEabrv,references}

\end{document}